\documentclass[11pt]{article}

\usepackage{amsmath,amssymb,amsthm,amsfonts,latexsym,bbm,xspace,graphicx,float,mathtools,mathdots,cancel}
\usepackage{braket,caption,subcaption,ellipsis,xcolor,textcomp,MnSymbol}

\usepackage[colorlinks,citecolor=blue,bookmarks=true]{hyperref}
\usepackage[nameinlink]{cleveref}

\usepackage{url}
\usepackage{bibentry} 
\usepackage[letterpaper,margin=1in]{geometry}

\usepackage[explicit]{titlesec}
\usepackage{bibentry}
\newtheorem{theorem}{Theorem}[section]
\newtheorem{lemma}[theorem]{Lemma}
\newtheorem{definition}[theorem]{Definition}

\newcommand\RS{{\mathrm{RS}}}
\newcommand\ORS{{\mathrm{ORS}}}

\newcommand\N{{\mathbb{N}}}

\newcommand\remove[1]{{}}

\newcommand{\ignore}[1]{}
\newcommand{\nobibentry}[1]{{\let\nocite\ignore\bibentry{#1}}}

\title{A note on Ordered Ruzsa-Szemer\'edi graphs}
\author{Kevin Pratt\thanks{Department of Computer Science, Courant Institute of Mathematical Sciences, New York University.}}

\begin{document}
	
	\maketitle
	
	\begin{abstract}
		A recent breakthrough of Behnezhad and Ghafari \cite{bg} and subsequent work of Assadi, Khanna, and Kiss \cite{assadi2025improved} gave algorithms for the fully dynamic $(1-\varepsilon)$-approximate maximum matching problem whose runtimes are determined by a purely combinatorial quantity: the maximum density of \emph{Ordered Ruzsa-Szemer\'edi} (ORS) graphs. We say a graph $G$ is an $(r,t)$-ORS graph if its edges can be partitioned into $t$ matchings $M_1,M_2, \ldots, M_t$ each of size $r$, such that for every $i$, $M_i$ is an induced matching in the subgraph $M_{i} \cup M_{i+1} \cup \cdots \cup M_t$. This is a relaxation of the extensively-studied notion of a Ruzsa-Szemer\'edi (RS) graph, the difference being that in an RS graph each $M_i$ must be an induced matching in $G$.
		
		In this note, we show that these two notions are roughly equivalent. Specifically, let $\ORS(n)$ be the largest $t$ such that there exists an $n$-vertex ORS-$(\Omega(n), t)$ graph, and define $\RS(n)$ analogously. We show that if $\ORS(n) \ge \Omega(n^c)$, then for any fixed $\delta > 0$, $\RS(n) \ge \Omega(n^{c(1-\delta)})$. This resolves a question of \cite{bg}. 
	\end{abstract}
	
	\section{Introduction}
	In the fully dynamic approximate matching problem, we are given an $n$-vertex graph $G$ which is undergoing edge insertions and deletions, and we seek to maintain a $(1-\varepsilon)$-approximate maximum matching after each update. Despite much work, the update complexity of this problem remains poorly understood. It is known that for $\varepsilon = 1/2$ one needs only constant update time \cite{solomon2016fully}, and for $\varepsilon > 1/3$ one can achieve $n^{1/2+o(1)}$ update time \cite{bernstein2016faster}. However, for any $\varepsilon \le 1/3$, it is only known that one can achieve $n^{1-o(1)}$ update time \cite{omv}.
	
	In \cite{bg}, the following notion was introduced in the context of this problem:
	\begin{definition}
	A graph $G$ is an $(r,t)$-ordered Ruzsa-Szemer\'edi $(\ORS)$ graph if its edge set can be partitioned into $t$ edge-disjoint matchings $M_1, M_2, \ldots, M_t$ each of size $r$, such that the subgraph of $G$ induced by $M_i$ does not include any edges from $M_j$ for any $j > i$.
	
	We let $\ORS(n,r)$ denote the maximum $t$ for which there exists an $n$-vertex $(r,t)$-$\ORS$ graph.
	\end{definition}
	The relevance of this quantity is due to an algorithm of \cite{bg} with update time \[\sqrt{n^{1+\varepsilon} \cdot \ORS(n, \Theta_\varepsilon(n)}.\]  Currently it is only known that for fixed $\varepsilon < 1/4$
	\[n^{o(1)} < \ORS(n, \varepsilon n) < n^{1-o(1)},\]
	with the upper bound due to \cite{bg} and the lower bound due to \cite{fischer2002monotonicity}. Importantly, note that if the lower bound is optimal, then \cite{bg} gives an $n^{1/2 + O(\varepsilon)}$-time algorithm, significantly improving on the current record for small $\varepsilon$. In follow-up work \cite{assadi2025improved}, this was improved to $n^{o(1)} \cdot \ORS(n, \Theta_\varepsilon(n))$ update time \cite{assadi2025improved}, thereby giving an algorithm whose complexity entirely rests on $\ORS(n, \Theta_\varepsilon(n))$.
	
	A closely related but much older (see \cite{ruzsa1978triple}) notion is that of a Ruzsa-Szemer\'edi graph: 
	\begin{definition}
	A graph $G$ is an $(r,t)$ Ruzsa-Szemer\'edi $(\RS)$ graph if its edge set can be partitioned into $t$ edge-disjoint matchings $M_1, M_2, \ldots, M_t$ each of size $r$, such that the subgraph of $G$ induced by $M_i$ is a matching.

	We let $\RS(n,r)$ denote the maximum $t$ for which there exists an $n$-vertex $(r,t)$-$\RS$ graph.
	\end{definition}
	
	As for $\ORS$ graphs, it is only known that $n^{o(1)} \le \RS(n, \varepsilon n) \le n^{1-o(1)}$ for $\varepsilon < 1/4$ \cite{fischer2002monotonicity, fox2011new}.
	
	Clearly we have that $RS(n,r) \le ORS(n,r)$; could it be that $\ORS$ is significantly larger than $\RS$? We show that the answer is no:

	\begin{theorem}\label{main}
	If $\ORS(n, \varepsilon n) \ge \Omega(n^c)$, then for any fixed $\delta > 0$, $\RS(n, \Theta(\varepsilon^{1/\delta} n)) \ge \Omega(n^{c(1-\delta)})$.
	\end{theorem}

	Thus, understanding the maximum density of ORS graphs with linear-sized matchings is roughly equivalent to understanding the density of RS graphs with linear-sized matchings. While previously it was conceivable that $\RS=n^{o(1)}$ while $\ORS=n^{1-o(1)}$, our result shows that this not possible. In the context of the dynamic matching problem, this means that the result of \cite{assadi2025improved} can be stated instead as an algorithm with update time of 
	\[
	n^{o(1)} \cdot RS(n,\Theta_\varepsilon(n)). 
	\]
	Looking at this differently, if the update time of the algorithm of \cite{assadi2025improved} turns out to be polynomial, then not only are current constructions of $\ORS$ graphs far from optimal, but so are current constructions of $\RS$ graphs.
	\section{A lower bound on $\RS$ via $\ORS$}
	\begin{lemma}\label{aux}
		For all $n,r,k$, $\RS(n^k,r^k) \ge \ORS(n,r)^{k-1}/k$.
	\end{lemma}
	\begin{proof}
		Let $G = M_1 \cup \cdots \cup M_t$ be an $n$-vertex ORS-$(r,t)$ graph. For an edge $(u,v) \in E(G)$, we let $f(u,v)$ denote the index of the matching to which $(u,v)$ belongs.
		
		For $s \in \N$, consider the graph $H_s$ with $ V(H_s) = V(G)^k$ and
		\[E(H_s) = \{ ((u_1,u_2, \ldots, u_k), (v_1, v_2, \ldots, v_k)) \in V(H_s)^2 : \forall i \in [k]\hspace{0.1cm} (u_i, v_i) \in E(G),
		\sum_{i=1}^k f(u_i, v_i) = s\}.\]
		
		First, $H_s$ is a graph on $n^k$ vertices. For $a \in [t]^k$ with $\sum a_i = s$, let $M_a$ be the set of edges $(U,V) = ((u_1, \ldots, u_k), (v_1, \ldots v_k))$ in $H_s$ for which $f(u_i, v_i) = a_i$. Then the sets $M_a$ partition the edges of $H_s$, and $|M_a| = \prod_i |M_{a_i}| = r^k$. 
		
		We claim that the subgraph of $H_s$ induced by each $M_a$ is a matching. First note that $M_a$ is a matching, since if any vertex $U \in M_a$ had two distinct neighbors $V,W \in M_a$, there would exist some $i$ with $V_i \neq W_i$, and then $(U_i, V_i)$ and $(U_i, W_i)$ would be edges in $G$ both belonging to $M_{a_i}$, a contradiction. To see why it is an \emph{induced} matching, let $(U, V)$ be an edge in the induced subgraph. We know that for all $i$ it must be the case that $f(U_i, V_i) \le a_i$, since otherwise the subgraph of $G$ induced by $M_{a_i}$ would violate the ORS property. Also, $\sum f(U_i, V_i) = s = \sum a_i$. Together these imply that $f(U_i, V_i) = a_i$ for all $i$, and hence $(U,V) \in M_a$.
				
		The number of matchings in $H_s$ equals the number of solutions to $a_1 + \cdots + a_k = s$ with $a_i \in [t]$. By the pigeonhole principle, for some choice of $s$ this is at least $t^k/(kt) = t^{k-1}/k$.
	\end{proof}
	\begin{proof}[Proof of \Cref{main}]
		By \Cref{aux} and the assumption, \[\RS(n^k, (\varepsilon n)^k) \ge \ORS(n, \varepsilon n)^{k-1}/k \ge \Omega_k(n^{c(k-1)}).\]
		For all $N = n^k$, this shows that $\RS(N, \varepsilon^k N) \ge \Omega_k(N^{c(k-1)/k})$. By choosing $k = \lceil 1/\delta \rceil$ the statement holds for all sufficiently large such $N$. For $N$ that is not a $k$th power, we can apply the construction for $\lceil N^{1/k} \rceil^k$ and then delete an arbitrary subset of $\lceil N^{1/k} \rceil^k - N$ vertices, which will only shrink the size of a matching by $O_k(N^{(k-1)/k})$.
	\end{proof}
	We remark that the idea of \Cref{aux} is motivated by a trick used in the context of fast matrix multiplication which converts a relaxed notion of an induced matching in a hypergraph into a legitimate induced matching in its tensor powers; see for example \cite[Lemma 3.4]{blasiak2017cap}, \cite[Theorem 23]{christandl2022larger}.
	
	\section{Acknowledgments}
	I thank Sepehr Assadi for encouraging me to write this note, and for feedback on an earlier draft.
	\bibliographystyle{amsalpha}
	\bibliography{refs}

\newcommand{\etalchar}[1]{$^{#1}$}
\providecommand{\bysame}{\leavevmode\hbox to3em{\hrulefill}\thinspace}
\providecommand{\MR}{\relax\ifhmode\unskip\space\fi MR }
\providecommand{\MRhref}[2]{%
  \href{http://www.ams.org/mathscinet-getitem?mr=#1}{#2}
}
\providecommand{\href}[2]{#2}
\begin{thebibliography}{BCC{\etalchar{+}}17}

\bibitem[AKK25]{assadi2025improved}
Sepehr Assadi, Sanjeev Khanna, and Peter Kiss, \emph{Improved bounds for fully
  dynamic matching via ordered ruzsa-szemeredi graphs}, Proceedings of the 2025
  Annual ACM-SIAM Symposium on Discrete Algorithms (SODA), SIAM, 2025,
  pp.~2971--2990.

\bibitem[BCC{\etalchar{+}}17]{blasiak2017cap}
Jonah Blasiak, Thomas Church, Henry Cohn, Joshua~A Grochow, Eric Naslund,
  William~F Sawin, and Chris Umans, \emph{On cap sets and the group-theoretic
  approach to matrix multiplication}, Discrete Analysis (2017).

\bibitem[BG24]{bg}
Soheil Behnezhad and Alma Ghafari, \emph{Fully dynamic matching and ordered
  ruzsa-szemerédi graphs}, 2024 IEEE 65th Annual Symposium on Foundations of
  Computer Science (FOCS), IEEE Computer Society, 2024, pp.~314--327.

\bibitem[BS16]{bernstein2016faster}
Aaron Bernstein and Cliff Stein, \emph{Faster fully dynamic matchings with
  small approximation ratios}, Proceedings of the twenty-seventh annual
  ACM-SIAM symposium on Discrete algorithms, SIAM, 2016, pp.~692--711.

\bibitem[CFTZ22]{christandl2022larger}
Matthias Christandl, Omar Fawzi, Hoang Ta, and Jeroen Zuiddam, \emph{Larger
  corner-free sets from combinatorial degenerations}, ITCS 2022-13th
  Innovations in Theoretical Computer Science Conference, 2022, pp.~1--2410.

\bibitem[FLN{\etalchar{+}}02]{fischer2002monotonicity}
Eldar Fischer, Eric Lehman, Ilan Newman, Sofya Raskhodnikova, Ronitt Rubinfeld,
  and Alex Samorodnitsky, \emph{Monotonicity testing over general poset
  domains}, Proceedings of the thiry-fourth annual ACM symposium on Theory of
  computing, 2002, pp.~474--483.

\bibitem[Fox11]{fox2011new}
Jacob Fox, \emph{A new proof of the graph removal lemma}, Annals of Mathematics
  (2011), 561--579.

\bibitem[Liu24]{omv}
Yang~P. Liu, \emph{{ On Approximate Fully-Dynamic Matching and Online
  Matrix-Vector Multiplication }}, 2024 IEEE 65th Annual Symposium on
  Foundations of Computer Science (FOCS) (Los Alamitos, CA, USA), IEEE Computer
  Society, October 2024, pp.~228--243.

\bibitem[RS78]{ruzsa1978triple}
Imre~Z Ruzsa and Endre Szemer{\'e}di, \emph{Triple systems with no six points
  carrying three triangles}, Combinatorics (Keszthely, 1976), Coll. Math. Soc.
  J. Bolyai \textbf{18} (1978), no.~939-945, 2.

\bibitem[Sol16]{solomon2016fully}
Shay Solomon, \emph{Fully dynamic maximal matching in constant update time},
  2016 IEEE 57th Annual Symposium on Foundations of Computer Science (FOCS),
  IEEE, 2016, pp.~325--334.

\end{thebibliography}
\end{document}